\documentclass[letterpaper, 10 pt, conference]{ieeeconf}
\IEEEoverridecommandlockouts                              
\usepackage{amsmath}
\usepackage{graphicx} 
 \usepackage{hyperref}
\usepackage{amssymb}  
\usepackage{algorithm,algorithmic}
\usepackage{mathtools}
\usepackage{siunitx}
\usepackage{tabularx}
\usepackage[usenames,dvipsnames]{xcolor}
\usepackage{pgfplots}
\usepgfplotslibrary{fillbetween}
\pgfplotsset{compat = newest}
\usetikzlibrary{arrows,positioning,shapes,intersections,patterns,calc,fit,external,decorations,decorations.markings} 
\usepackage{cite}

\newtheorem{rem}{Remark}
\newtheorem{thm}{Theorem}
\newtheorem{lem}{Lemma}

\newtheorem{propy}{Property}
\newtheorem{assum}{Assumption}

\newcommand\tran{\mkern-2mu\raise1.25ex\hbox{$\scriptscriptstyle\top\hspace{0.5mm}$}\mkern-3.5mu}
\newcommand{\R}{\mathbb{R}}

\newcommand{\N}{\mathbb{N}}

\newcommand{\D}{\mathcal{D}}
\newcommand{\X}{\mathcal{X}}

\newcommand{\bm}[1]{{\boldsymbol{#1}}}
\newcommand{\Verts}[1]{{\left\Vert #1 \right\Vert}}
\DeclareMathOperator{\diag}{diag}
\DeclareMathOperator{\var}{var}
\DeclareMathOperator{\mean}{\mu}
\DeclareMathOperator{\Var}{\Sigma}
\DeclareMathOperator{\Mean}{\bm\mu}

\DeclareMathOperator{\eig}{eig}

\DeclareMathOperator{\Prob}{P}

\newcommand{\z}{\bm z}

\newcommand{\x}{\bm x}

\newcommand{\dx}{\dot{\bm x}}

\newcommand{\f}{\bm{f}}
\newcommand{\g}{\bm{g}}

\newcommand{\y}{\bm{y}}

\usepackage[noabbrev]{cleveref} 
\crefname{rem}{Remark}{Remarks}
\crefname{exam}{Example}{Examples}
\crefname{assum}{Assumption}{Assumptions}
\crefname{prop}{Proposition}{Propositions}
\crefname{propy}{Property}{Properties}
\crefname{cor}{Corollary}{Corollaries}
\crefname{lem}{Lemma}{Lemmas}
\crefname{section}{Section}{Sections}
\crefname{thm}{Theorem}{Theorems}
\crefname{defn}{Definition}{Definitions}
\crefname{figure}{Fig.}{Fig.}
\Crefname{figure}{Figure}{Figures}
\crefname{equation}{}{}

\title{\LARGE \bf Online learning-based trajectory tracking for underactuated vehicles with uncertain dynamics}

\author{Thomas Beckers$^{1}$, Leonardo J. Colombo$^{2}$, Sandra Hirche$^{3}$, George J. Pappas$^{1}$%
\thanks{$^{1}$ are with the Department of Electrical and Systems Engineering, University of Pennsylvania, Philadelphia, PA 19104, USA, {\tt\small \{tbeckers,pappasg\}@seas.upenn.edu}}%
\thanks{$^{2}$ is with the Instituto de Ciencias Matematicas (CSIC-UAM-UCM-UC3M), Calle Nicolas Cabrera 13-15, Campus Cantoblanco, Madrid, 28049, Spain, {\tt\small leo.colombo@icmat.es}}
\thanks{$^{3}$ is with the Chair of Information-oriented Control, Technical University of Munich, Munich, 80333, Germany, {\tt\small hirche@tum.de}}%
}

\begin{document}

\maketitle
\thispagestyle{empty}
\pagestyle{empty}

\begin{abstract}
Underactuated vehicles have gained much attention in the recent years due to the increasing amount of aerial and underwater vehicles as well as nanosatellites. Trajectory tracking control of these vehicles is a substantial aspect for an increasing range of application domains. However, external disturbances and parts of the internal dynamics are often unknown or very time-consuming to model. To overcome this issue, we present a tracking control law for underactuated rigid-body dynamics using an online learning-based oracle for the prediction of the unknown dynamics. We show that Gaussian process models are of particular interest for the role of the oracle. The presented approach guarantees a bounded tracking error with high probability where the bound is explicitly given. A numerical example highlights the effectiveness of the proposed control law.
\end{abstract}

\section{Introduction}
The demand for unmanned aerial and underwater vehicles is rapidly increasing in many areas such as  monitoring, mapping, agriculture, and delivery. These vehicles are typically underactuated due to constructional reasons which poses several challenges from the control perspective~\cite{reyhanoglu1999dynamics}. The dynamics of these systems can often be expressed by rigid bodies motion with full attitude control and one translational force input. This is a classical problem in underactuated mechanics and many different types of control methods have been proposed to achieve an accurate trajectory tracking. Most of the control approaches are mainly based on feedback linearization~\cite{lee2009feedback} and backstepping methods~\cite{raffo2008backstepping} which are also analyzed in terms of stability, e.g., in~\cite{frazzoli2000trajectory}.

However, these control approaches depend on exact models of the systems and possible external disturbances to guarantee stability and precise tracking. An accurate model of typical uncertainties is hard to obtain by using first principles based techniques. Especially the impact of air/water flow on aerial/underwater vehicles or the interaction with unstructured and a-priori unknown environment further compound the uncertainty. The increase of the feedback gains to suppress the unknown dynamics is unfavorable due to the large errors in the presence of noise and the saturation of actuators. A suitable approach to avoid the time-consuming or even unfeasible modeling process is provided by learning-based oracles such as neural networks or Gaussian processes~(GPs). These data-driven modeling tools have shown remarkable results in many different applications including control, machine learning and system identification. In data-driven control, data of the unknown system dynamics is collected and used by the oracle to predict the dynamics in areas without training data. In contrast to parametric models, data-driven models are highly flexible and are able to reproduce a large class of different dynamics, see~\cite{hou2013model}.

The purpose of this article is to employ the power of online learning-based approaches for the tracking control for a class of underactuated systems. At the end, stability and a desired level of performance of the closed-loop system should be guaranteed. The problem of tracking control of underactuated aerial/underwater vehicles with uncertainties has been addressed in~\cite{mahony2004robust,kobilarov2013trajectory} but these approaches are restricted to structured uncertainties such as uncertain parameters or use high feedback gains for compensation. Safe feedback linearization based on GPs are introduced in~\cite{umlauft:TAC2020,greef2021letters} for a specific class of systems but they do not capture the general underactuated nature of the here considered model class and are limited to single input systems. 

Learning-based approaches for Euler-Lagrange systems with stability guarantees are presented in~\cite{beckers2019automatica,helwa-ral19,han2021stable}. However, the systems are required to be fully actuated or are limited to the class of balancing robots. For a specific type of aerial vehicles, a safe Gaussian process based controller is proposed in~\cite{berkenkamp-icra16} but with additional assumptions such as an initial safe controller. The contribution of this article is a online learning-based tracking control law for a large class of underactuated vehicles with stability and performance guarantees. Instead of focusing on a particular type of oracle, the proposed approach allows the usage of various learning-based oracles. The online learning approach allows to improve the model and, thus, the tracking performance during runtime.

The remaining article is structured as follows: After the problem setting in~\cref{sec:ps}, the online learning-based oracle and the tracking controller are introduced in~\cref{sec:lbc}. Finally, a numerical example is presented in~\cref{sec:num}.

\section{Problem Setting}\label{sec:ps}
We assume a single underactuated rigid body with position\footnote{Vectors are denoted with bold characters and matrices with capital letters. The term~$A_{i,:}$ denotes the i-th row of the matrix~$A$. The expression~$\mathcal{N}(\mu,\Sigma)$ describes a normal distribution with mean~$\mu$ and covariance~$\Sigma$. The probability function is denoted by $\Prob$. The set $\R_{>0}$ denotes the set of positive real numbers.} $\bm{p}\in\R^3$ and orientation matrix $R\in SO(3)$. The body-fixed angular velocity is denoted by $\bm{\omega}\in\R^3$. The vehicle has mass $m\in\R_{>0}$ and rotational inertia tensor $J\in\R^{3\times 3}$. The state space of the vehicle is $S=SE(3)\times\R^6$ with $\bm{s}=((R,\bm{p}),(\bm{\omega},\dot{\bm{p}}))\in S$ denoting the whole state of the system. The vehicle is actuated with control torques $\bm{\tau}\in\R^3$ and a control force $u\in\R$, which is applied in a body-fixed direction defined by a unit vector $\bm{e}\in\R^3$.
We can model the system as
\begin{align}
	\begin{split}
	m \ddot{\bm{p}}&=R\bm{e}u+\f(\bm{p},\dot{\bm{p}})\\
	\dot{R}&=R\check{\bm{\omega}}\\
	\dot{\bm{\omega}}&=J^{-1}\big( J \bm{\omega}\times \bm{\omega}+\bm{\tau}+\f_\omega(\bm{s})\big),\label{for:system}
	\end{split}
\end{align}
where the map $\check{(\cdot)}\colon\R^3\to \mathfrak{so}(3)$ is given by
\begin{align}
	\check{\bm{\omega}}=\begin{bmatrix}
	0 & -\omega_3 & -\omega_2\\
	\omega_3 & 0 & -\omega_1\\
	-\omega_2 & \omega_1 & 0
	\end{bmatrix},
\end{align}
with the components of the angular velocity $\bm{\omega}=[\omega_1,\omega_2,\omega_3]^\top$. The functions $\f\colon\R^6\to\R^3$ and $\f_\omega\colon S\to\R^3$ are state-depended disturbances and/or unmodeled dynamics. It is assumed that the full state $\bm{s}$ can be measured. The general objective is to track a trajectory specified by the functions $(R_d,\bm{p}_d)\colon [0,T]\to SE(3)$. For simplicity, we focus here on position tracking only. The extension to rotation tracking is straightforward and will be discussed later.
\subsection{Equivalent system}
In preparation for the learning and control step, we transform the system dynamics~\cref{for:system} in an equivalent form. For the unknown dynamics $\f$ and $\f_\omega$, we use the estimates $\hat{\f}\colon\R^6\to\R^3$ and $\hat{\f}_\omega\colon S\to\R^3$, respectively, of an oracle. The estimation error is moved to $\bm{\rho}(\x)=\f(\x)-\hat{\f}(\x)$ and $\bm{\rho}_\omega(\bm{s})=\f_\omega(\bm{s})-\hat{\f}_\omega(\bm{s})$. With the system matrix $A\in\R^{6\times 6}$ and input matrix $B\in\R^{6\times 3}$ given by
\begin{align}
	A=\begin{bmatrix}
	0 & I_3\\ 0 & 0
	\end{bmatrix},\quad B=\begin{bmatrix}
	0 \\ \frac{1}{m}I_3
	\end{bmatrix},
\end{align}
and $I_3\in\R^{3\times 3}$ as identity matrix, we can rewrite~\cref{for:system} as 
\begin{align}
	\begin{split}
	\dx&=A\x+B\big(\g(R,u)+\hat{\f}(\x)+\bm{\rho}(\x)\big)\\
	\dot{R}&=R\check{\bm{\omega}}\\
	\dot{\bm{\omega}}&=J^{-1}\big( J \bm{\omega}\times \bm{\omega}+\bm{\tau}+\hat{\f}_\omega(\bm{s})+\bm{\rho}_\omega(\bm{s})\big),\label{for:system1}
	\end{split}
\end{align}
where $\x=[\bm{p}^\top,\dot{\bm{p}}^\top]^\top\in\R^6$, $\bm{s}\in S$ and $\bm{g}\colon SO(3)\times \R\to\R^3$ is a virtual control input with $\bm{g}(R,u)\coloneqq R\bm{e}u$. As consequence,~\cref{for:system1} is equivalent to~\cref{for:system} without loss of generality. 
\section{Learning-based control}\label{sec:lbc}
An overview about the proposed control strategy is depicted in~\cref{fig:block}. As introduced in the vehicle's dynamics~\cref{for:system1}, we assume that parts of the dynamics are known and some are unknown, i.e., $\f$ and $\f_\omega$. The proposed control strategy is based on a backstepping controller using an internal model that is updated by the predictions of an oracle. The data for the oracle is collected in arbitrary time intervals of the vehicle's dynamics during the control process. Then, the predictions of the oracle are updated based on the collected dataset and the vehicle model is improved. In the following, we present the online learning and control part in detail.
\begin{figure}[t]
\begin{center}
\tikzsetnextfilename{block}
	\begin{tikzpicture}[auto, node distance=2cm,>=latex']
	\tikzstyle{block} = [draw, fill=white, rectangle,  line width=1pt, 
    minimum height=2.5em, minimum width=5em, font=\small,align=center]
    \tikzstyle{mux} = [draw, fill=black, rectangle,  line width=1pt, 
    minimum height=2.5em,inner sep=1pt, minimum width=0.01cm]
	\tikzstyle{sum} = [draw, fill=white, circle, node distance=1cm]
	\tikzstyle{circ} = [draw=red, fill=white, circle, inner sep=1pt, node distance=1cm]
	\tikzstyle{cross} = [draw, fill=black, circle,inner sep=1pt, node distance=1cm]
	\tikzstyle{input} = [coordinate]
	\tikzstyle{output} = [coordinate]
	\tikzstyle{mid} = [coordinate]

    \node [input, name=input] {};
    \node [sum, name=midp1,right of=input, node distance=0.8cm,label={[label distance=-0.01cm]200:$-$}] {};
    \node [mid, below of=midp1, node distance=2cm] (midp4) {};
    \node [block, right of=input, node distance=2.5cm] (controller) {Backstepping};
    \node [block, draw=red, above of=controller, node distance=1.5cm] (vmodel) {Vehicle\\model};
    \node [coordinate, above left=0.12cm and 0.27cm of vmodel, node distance=1.2cm] (f2) {};
    \node [coordinate, below right=0.12cm and 0.27cm of vmodel, node distance=1.2cm] (f1) {};
    \node [block, right of=controller, node distance=3.3cm] (knownd) {Known\\dynamics};
    \node [mid, right of=knownd, node distance=1.1cm] (midp2) {};
    \node [mux, right of=midp2, node distance=0.5cm] (midp3) {};
    \node [block, below of=knownd, node distance=1.3cm] (unknownd) {Unknown\\dynamics};
    \node [block, draw=red, above of=knownd, node distance=1.5cm] (dataset) {Dataset};
    \node [block, draw=red, above of=dataset, node distance=1.5cm] (model) {Oracle};

    \draw [->] (input) -- node[pos=0.3] {$\x_d$} (midp1);
    \draw [->] (midp1) --  node[pos=0.5]  {$\bm{z}_0$} (controller);
    \draw [->] (controller) -- node[pos=0.3]  {$\bm{u},\bm{\tau}$} (knownd);
    \draw [->] (unknownd) -- node[pos=0.5,right]  {\small $\f,\f_\omega$} (knownd);
    \draw [] (knownd) -- (midp2) -- (midp3);
    \draw [->] (midp2) |- (unknownd);
    \draw [->,red] ([yshift=0.1cm] midp3.east) -- ++(0.2cm,0) |- node[pos=0.1,right]  {$\bm{s}$} (model);
    \draw [red] (model.west) -| node[pos=0.3, above]  {$\hat{\f}(\x),\hat{\f}_\omega(\bm{s})$} (f2);
    \draw [->,red] (f2) -- (f1);
    \draw [->,red] (vmodel.south) -- (controller.north);
    \draw [] ([yshift=-0.1cm] midp3.east) -- ++(0.2cm,0) |- node[pos=0.15]  {$\bm{x}$} (midp4);
    \draw [->] (midp4) -- (midp1);
    
    \node[draw,dashed,line width=1pt,xshift=0.05cm,inner xsep=0.3cm,inner ysep=0.05cm,fit=(knownd) (unknownd),label={[label distance=-0.01cm]190:Vehicle}] (box) {};
    \node [circ, above of=box, node distance=1.3cm,xshift=-0.05cm] (circ1) {};
    \node [circ, below of=dataset, node distance=0.6cm] (circ2) {};
    \draw [red] ([xshift=-0.05cm]box.north) -- (circ1.south);
    \draw [red] (circ1.north) -- node[pos=0.5,left]  {\small Collection} ([xshift=0.04cm]circ2.east);
    \draw [red] (circ2.north) -- (dataset.south);
    \node [circ, above of=dataset, node distance=0.6cm] (circ3) {};
    \node [circ, below of=model, node distance=0.6cm] (circ4) {};
    \draw [red] (dataset.north) -- (circ3.south);
    \draw [red] (circ3.north) -- node[pos=0.4,left]  {\small Oracle update} ([xshift=0.04cm]circ4.east);
    \draw [red] (circ4.north) -- (model.south);
    
\end{tikzpicture}
	\vspace{-0.2cm}
	\caption{Block diagram of the proposed control law.}
	\vspace{-0.7cm}
	\label{fig:block}
\end{center}
\end{figure}
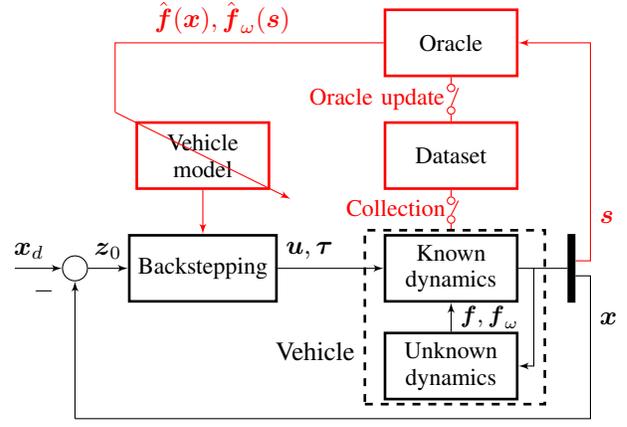
\subsection{Learning}
For the learning of the unknown dynamics of~\cref{for:system}, we consider an oracle which predicts the values of $\f,\f_\omega$ for a given state $\bm{s}$. For this purpose, the oracle collects $N(n)\colon\N\to\N$ training points of the system~\cref{for:system} to create a data set
\begin{align}
\D_{n(t)}=\{\bm{s}^{\{i\}},\y^{\{i\}}\}_{i=1}^{N(n)}.\label{for:dataset}
\end{align}
The output data $\y\in\R^6$ are given by $\y=[(m\ddot{\bm{p}}-R\bm{e}u)^\top,(J(\dot{\bm{\omega}}-\bm{\omega}\times \bm{\omega})-\bm{\tau})^\top]^\top$ such that the first three components of $\y$ correspond to $\f$ and the remaining to $\f_\omega$. The data set $\D_{n(t)}$ with $n\colon \R_{\geq 0}\to\N$ can change over time $t$, such that at time $t_1\in\R_{\geq 0}$ the data set $\D_{n(t_1)}$ with $N(n(t_1))$ training points exists. This allows to accumulate training data over time, i.e., $N(n)$ is monotonically increasing, but also "forgetting" of training data to keep $N(n)$ constant. The time-dependent estimates of the oracle is denoted by $\hat{\f}_n(\x)$ and $\hat{\f}_{\omega,n}(\bm{s})$ to highlight the dependence on the corresponding data set $\D_{n}$. Note that this construction also allows offline learning, i.e. the prediction of the oracle depends on previous collected data only, or any hybrid online/offline approach.
\begin{rem}
Simple oracles can be parametric models such as a linear model, where the parameters are learned with a least-square approach based on the data set $\D_n$. More powerful oracles 
are given by neural networks, due to their universal function approximation property~\cite{scarselli1998universal}. Furthermore, non-parametric oracles such as Gaussian processes and support vector machines have led to promising results as probabilistic function approximators~\cite{rasmussen2006gaussian,steinwart2008support}.
\end{rem}
For the later stability analysis of the closed-loop, we introduce the following assumptions, which cover various types of oracles.
\begin{assum}\label{ass:1}
	Consider an oracle with the predictions $\hat{\f}_n\in\mathcal{C}^2$  and $\hat{\f}_{\omega,n}\in\mathcal{C}^0$ based on the data set $\D_n$~\cref{for:dataset}. Let $S_\X\subset (SE(3)\times (\X\subset\R^6))$ be a compact set where the derivatives of $\hat{\f}_n$ are bounded on $\X$. There exists a bounded function $\bar{\rho}_n\colon S_\X\to\R_{\geq 0}$ such that the prediction error is given by
	\begin{align}
		\Prob\Bigg\lbrace\left\Vert\begin{bmatrix} \f(\x)- \hat{\f}_n(\x) \\ \f_\omega(\bm{s})- \hat{\f}_{\omega,n}(\bm{s})\end{bmatrix}\right\Vert\leq \bar{\rho}_n(\bm{s})\Bigg\rbrace\geq \delta
	\end{align}
	with probability $\delta\in (0,1]$ for all $\x\in\X,\bm{s}\in S_\X$ and $n(t)$.
\end{assum}
\begin{assum}\label{ass:2}
The number of data sets $\D_n$ is finite and there are only finitely many switches of~$n(t)$ over time, such that there exists a time $T\in\R_{\geq 0}$ where $n(t)=n_{\text{end}},\forall t\geq T$
\end{assum}
\Cref{ass:1} is fulfilled, for instance, by a Gaussian process model as oracle as shown in the next section. The second assumption is little restrictive since the number of sets is often naturally bounded due to finite computational power or memory limitations and since the unknown functions $\f,\f_\omega$ in~\cref{for:system} is not time-dependent, long-life learning is typically not required. Furthermore,~\cref{ass:2} ensures that the switching between the data sets is not infinitely fast which is natural in real world applications.
\subsection{Gaussian process as oracle}
Gaussian process models have been proven as very powerful oracle for nonlinear function regression. For the prediction, we concatenate the~$N(n)$ training points of $\D_n$ in an input matrix~$X=[\bm{s}^1,\bm{s}^2,\ldots,\bm{s}^{N(n)}]$ and a matrix of outputs~$Y^\top=[\y^1,\y^2,\ldots,\y^{N(n)}]$, where $\y$ might be corrupted by additive Gaussian noise with $\mathcal{N}(0,\sigma I_6)$. Then, a prediction for the output $\y^*\in\R^6$ at a new test point $\bm{s}^*\in S_\X$ is given by
\begin{align}
	\mean_i(\y^*\vert \bm{s}^*,\D_n)&=m_i(\bm{s}^*)+\bm{k}(\bm{s}^*,X)^\top K^{-1}\label{for:gppred}\\
	&\phantom{=}\left(Y_{:,i}-[m_i(X_{:,1}),\ldots,m_i(X_{:,N})]^\top\right)\notag\\
	\var_i(\y^*\vert \bm{s}^*,\D_n)&=k(\bm{s}^*,\bm{s}^*)-\bm{k}(\bm{s}^*,X)^\top K^{-1} \bm{k}(\bm{s}^*,X).\notag
\end{align}
for all $i\in\{1,\ldots,6\}$, where $Y_{:,i}$ denotes the $i$-th column of the matrix of outputs~$Y$. The kernel $k\colon S_\X\times S_\X\to\R$ is a measure for the correlation of two states~$(\bm{s},\bm{s}^\prime)$, whereas the mean function $m_i\colon S_\X\to\R$ allows to include prior knowledge. The function~$K\colon S_\X^N\times  S_\X^N\to\R^{N\times N}$ is called the Gram matrix whose elements are given by $K_{j',j}= k(X_{:, j'},X_{:, j})+\delta(j,j')\sigma^2$ for all $j',j\in\{1,\ldots,N\}$ with the delta function $\delta(j,j')=1$ for $j=j'$ and zero, otherwise. The vector-valued function~$\bm{k}\colon S_\X\times  S_\X^N\to\R^N$, with the elements~$k_j = k(\bm{s}^*,X_{:, j})$ for all $j\in\{1,\ldots,N\}$, expresses the covariance between~$\bm{s}^*$ and the input training data $X$. The selection of the kernel and the determination of the corresponding hyperparameters can be seen as degrees of freedom of the regression. A powerful kernel for GP models of physical systems is the squared exponential kernel. An overview about the properties of different kernels can be found in~\cite{rasmussen2006gaussian}. As we use the oracle in an online setting where new training data is collected over time, the dataset $\D_n$ for the prediction~\cref{for:gppred} changes over time. The GP model allows to integrate new training data in a simple way by exploiting that every subset follows a multivariate Gaussian distribution, see~\cite{rasmussen2006gaussian} for more details.
\begin{rem}
	 The mean function can be achieved by common system identification techniques of the unknown dynamics~$\f,\f_\omega$ as described in~\cite{aastrom1971system}. However, without any prior knowledge the mean function is set to zero, i.e. $m_i(\bm{s})=0$.
\end{rem}
 Based on~\cref{for:gppred}, the normal distributed components~$y^*_i\vert \bm{s}^*,\D_n$ are combined into a multi-variable distribution which leads to $\y^*\vert(\bm{s}^*,\D_n) \sim \mathcal{N} (\bm\mean(\cdot),\Var(\cdot))$, where
\begin{align}
\begin{split}
	\bm \mean(\y^*\vert \bm{s}^*,\D_n)&=[\mean_1(\cdot),\ldots,\mean_{6}(\cdot)]^\top\\
	\Var(\y^*\vert \bm{s}^*,\D_n)&=\diag\left[\var_1(\cdot),\ldots,\var_{6}(\cdot)\right].
\end{split}
\end{align} 
\begin{rem}
For notational simplicity, we consider identical kernels for each output dimension. However, the GP model can be easily adapted to different kernels for each output dimension.
\end{rem}
With the introduced GP model, we are now addressing~\cref{ass:1} using~\cite{steinwart2008support,beckers2019automatica,beckers:ecc2016}. To provide model error bounds, additional assumptions on the unknown parts of the dynamics~\cref{for:system} must be introduced, in line with the no-free-lunch theorem, see~\cite{wolpert1996lack}. 
\begin{assum}\label{as:rkhs}
	The kernel~$k$ is selected such that~$\f,\f_\omega$ have a bounded reproducing kernel Hilbert space (RKHS) norm on~$\X$ and $S_\X$, respectively, i.e.~$\Verts{f_i}_k<\infty,\Verts{f_{\omega,i}}_k<\infty$ for all~$i=1,2,3$.
\end{assum}
The norm of a function in a RKHS is a smoothness measure relative to a kernel~$k$ that is uniquely connected with this RKHS. In particular, it is a Lipschitz constant with respect to the metric of the used kernel. A more detailed discussion about RKHS norms is given in~\cite{wahba1990spline}.~\Cref{as:rkhs} requires that the kernel must be selected in such a way that the functions~$\f,\f_\omega$ are elements of the associated RKHS. This sounds paradoxical since this function is unknown. However, there exist some kernels, namely universal kernels, which can approximate any continuous function arbitrarily precisely on a compact set~\cite[Lemma 4.55]{steinwart2008support} such that the bounded RKHS norm is a mild assumption. Finally, with~\cref{as:rkhs}, the model error can be bounded as written in the following lemma.
\begin{lem}[adapted from~\cite{beckers2019automatica}]
	\label{lemma:boundederror}
	Consider the unknown functions $\f,\f_\omega$ and a GP model satisfying~\cref{as:rkhs}. The model error is bounded by
	\begin{align}
		\Prob\Bigg\lbrace\Bigg\Vert & \Mean\Bigg(\begin{bmatrix}\hat{\f}_n(\x)\\ \hat{\f}_{\omega,n}(\bm{s})\end{bmatrix}\Bigg\vert \bm{s},\D_n\Bigg)-\begin{bmatrix}\f(\x)\\\f_\omega(\bm{s})\end{bmatrix}\Bigg\Vert\notag\\
		&\leq \Bigg\Vert\bm \beta\tran \Var^{\frac{1}{2}}\Bigg(\begin{bmatrix}\hat{\f}_n(\x)\\ \hat{\f}_{\omega,n}(\bm{s})\end{bmatrix}\Bigg\vert \bm{s},\D_n\Bigg)\Bigg\Vert\Bigg\rbrace\geq \delta\notag
	\end{align}
	for~$\x\in\X,\bm{s}\in S_\X,\delta\in(0,1)$ with $\bm\beta\in\R^6$ given by~\cite[Lemma 1]{beckers2019automatica}
\end{lem}
\begin{proof}
It is a direct implication of~\cite[Lemma 1]{beckers2019automatica}.
\end{proof}
With~\cref{as:rkhs} and the fact, that universals kernels exist which generate bounded predictions with bounded derivatives, see~\cite{beckers:ecc2016}, GP models can be used as oracle to fulfill~\cref{ass:1}. In this case, the prediction error bound is given by $\bar{\rho}_n(\bm{s})\coloneqq \Vert\bm \beta\tran \Var^{\frac{1}{2}}([\hat{\f}_n(\x)^\top,\hat{\f}_{\omega,n}(\bm{s})^\top]^\top\vert \bm{s},\D_n)\Vert$ as shown in~\cref{lemma:boundederror}.
\begin{rem}
An efficient greedy algorithm can be used to find $\bm{\beta}$ based on the maximum information gain~\cite{srinivas2012information}.
\end{rem}
\subsection{Tracking control}
For the tracking control, we consider a given desired trajectory $\x_d(t)\colon\R_{t\geq 0}\to\X,\x_d\in\mathcal{C}^4$. The tracking error is denoted by $\z_0(t)=\x(t)-\x_d(t)$.  Before we propose the main theorem about the safe learning-based tracking control law, the feedback gain matrix $G_n$ is introduced. As part of the controller, $G_n$ penalizes the position tracking error and the result is fed back to both inputs, the force control $u$ and the torque control $\bm{\tau}$ of the system~\cref{for:system}. The feedback gain matrix is allowed to be adapted with any update of the oracle based on a new data set $\D_n$ to lower the feedback gains when the oracle's accuracy is improved. 
\begin{propy}
\label{ass:Lyap}
The matrix $G_n\in\R^{3\times 6}$ is chosen such that there exist a symmetric positive definite matrix $P_n\in\R^{6\times 6}$ and a positive definite matrix $Q_n\in\R^{6\times 6}$ which satisfy the Lyapunov equation
\begin{align}
	P_n\big(A-BG_n\big)+\big(A-BG_n\big)^\top P_n= -Q_n
\end{align}
for each switch of $n(t)$.
\end{propy}
\Cref{ass:Lyap} is satisfied if the real parts of all eigenvalues of $(A-BG_n)$ are negative. For example, this can be achieved by any $G_n=[G_{n,1},G_{n,2}]$, where $G_{n,1},G_{n,2}\in\R^{3\times 3}$ are positive definite diagonal matrices.
\begin{thm}
\label{thm:1}
Consider the underactuated rigid-body system given by~\cref{for:system} with unknown dynamics $\f,\f_\omega$ and the existence of an oracle satisfying~\cref{ass:1,ass:2}. Let $G_{z_1},G_{z_2}\in\R^{3\times 3}$ be positive definite symmetric matrices. With~\cref{ass:Lyap}, the control law
\begin{align}
	\bm{\tau}&=J(\bm{e}\!\times\!(R^\top\! \bm{g}_{\ddot{d}}-\check{\bm\omega}^2\bm{e}u-2\check{\bm\omega}\bm{e}\dot{u})u^{-1})\!-\! J\bm{\omega}\!\times\!\bm{\omega}\! -\! \hat{\f}_\omega(\bm{s}),\notag\\
	\ddot{u}&=\bm{e}^\top(R^\top \bm{g}_{\ddot{d}}-\check{\bm\omega}^2\bm{e}u-2\check{\bm\omega}\bm{e}\dot{u}),\label{for:ctrllaw}
\end{align}
with the desired virtual control input derivative	
\begin{align}	
	\g_{\ddot{d}}&=m\bm{p}_d^{(4)}-G_n\left(\frac{\partial \dot{\hat{\x}}}{\partial \x}\dot{\hat{\x}}-\ddot{\x}_d\right)-BP_n(\dot{\hat{\x}}-\dot{\x}_d)\notag\\
&-(G_{z_1}+G_{z_2})\left(\dot{\g}-m\bm{p}_d^{(3)}+G_n(\dot{\hat{\x}}-\dot{\x}_d)+\frac{\partial \hat{\f}_n}{\partial \x}\dot{\hat{\x}}\right)\notag\\
&-(G_{z_2}G_{z_1}+I_3)\left(\bm{g}-m \ddot{\bm{p}}_d+G_n\z_0+\hat{\f}_n(\x)\right)\notag\\
&-G_{z_2}B^\top P_n \z_0-\frac{\partial}{\partial \x}\Big[\frac{\partial{\hat{\f}}_n}{\partial \x}\dot{\hat{\x}}\Big]\dot{\hat{\x}}.\label{for:gddd}\\
\dot{\hat{\x}}&=A\x+B\left(\g(R,u)+\hat{\f}_n(\x)\right)\label{for:hatx}
\end{align}
guarantees that the tracking error is uniformly ultimately bounded in probability by 
\begin{align}
\Prob\{\Verts{\z_0(t)}\leq \max_{\bm{s}\in S_\X}\bar{\rho}_{n_\text{end}}(\bm{s})b_{n_\text{end}},\forall t\geq T\}\geq\delta
\end{align}
with $b_{n_\text{end}}=(\max\{\eig(P_{n_\text{end}}),1\}/\min\{\eig(P_{n_\text{end}}),1\})^{1/2}$, and time constant $T\in\R_{\geq 0}$ on $S_\X$.
\end{thm}
\begin{rem}
The control law does not depend on any state derivatives, which are typically noisy in measurements. The derivatives, i.e. the translational and angular accelerations, are only necessary for the training of the oracle, see~\cref{for:dataset}, which can often deal with noisy data. For instance, GP models can handle additive Gaussian noise on the output~\cite{rasmussen2006gaussian}.
\end{rem}
We prove the stability of the closed-loop with the proposed control law with multiple Lyapunov function, where the $n$-th function is active when the oracle predicts based on the corresponding training set $\D_n$. Note that due to a finite number of switching events, the switching between stable systems can not lead to an unbounded trajectory, see~\cite{liberzon1999basic}.\\
\begin{proof}
The term $\g(R,\bm{u})$ in~\cref{for:system1} is assumed as virtual control input with the desired force
\begin{align}
	\g_d(t,n,\x)=m \ddot{\bm{p}}_d-G_n\z_0-\hat{\f}_n(\x)\label{for:gd}
\end{align}
where $G_n$ can change by the switching of $n(t)$. The tracking error dynamics are given by
\begin{align}
	\dot{\z}_0&=A\x+B\left(\g(R,u)+\hat{\f}_n(\x)+\bm{\rho}_n(\x)\right)-\begin{bmatrix}
	\dot{\bm{p}}_d\\ \ddot{\bm{p}}_d
	\end{bmatrix}.\label{for:err_dyn}
\end{align}
Using the desired acceleration $\ddot{\bm{p}}_d$ of~\cref{for:gd} in~\cref{for:err_dyn} leads to
\begin{align}
	\dot{\z}_0&=\big(A-BG_n\big)\z_0+B\big(\g(R,u)-\g_d(t,n,\x)+\bm{\rho}_n(\x)\big).\notag
\end{align}
In the next step, the boundedness of the tracking error $\z_0$ is proven. For this purpose, we use the matrices $P_n,Q_n$ of~\cref{ass:Lyap} to construct the Lyapunov function~$V_{0,n}(\z_0)=0.5\z_0^\top P_n \z_0$ and compute its evolution
\begin{align}
	\dot{V}_{0,n}&= -\z_0^\top Q_n \z_0+(B^\top P_n \z_0)^\top \big(\g(R,u)-\g_d+\bm{\rho}_n(\x)\big).\notag
\end{align}
The first summand is negative for all $\z_0\in\R^6$. In the next step, we extend the previous Lyapunov function with the error term $\z_1\in\R^3$ with $\bm{z}_1(t,n,\x,R,u)=\g(R,u)-\g_d(t,n,\x)$, which describes the error between the virtual and the desired control input. Thus, it leads to a switching Lyapunov function $V_{1,n}(\z_0,\z_1)=V_{0,n}+0.5\z_1^\top\z_1 \geq 0$. The derivative of $V_{1,n}$ leads to
\begin{align}
\dot{V}_{1,n}=\dot{V}_{0,n}+\bm{z}_1^\top\left(\dot{\g}-m\bm{p}_d^{(3)}+G_n\dot{\z}_0+\dot{\hat{\f}}_n(\x)\right),\label{for:dV1}
\end{align}
where $\bm{p}_d^{(3)}$ denotes the third time-derivative of the desired position $\bm{p}_d$. Following again the idea of a desired virtual input as in~\cref{for:gd}, we construct a desired value of $\dot{\g}$ with
\begin{align}
\g_{\dot{d}}=m\bm{p}_d^{(3)}-G_n(\dot{\hat{\x}}-\dot{\x}_d)-B^\top P_n\z_0-G_{z_1}\z_1-\frac{\partial{\hat{\f}}_n}{\partial \x}\dot{\hat{\x}}.\label{for:gdd}
\end{align}
Instead of having dependencies on the typical noisy state derivative $\dot{\x}$, we use the estimation $\dot{\hat{\x}}\in\R^6$ given by~\cref{for:hatx}, which only contains the known parts of the system dynamics~\cref{for:system1}. Then, the expression~\cref{for:gdd} is used to substitute $\dot{\bm{g}}$ in~\cref{for:dV1}. This leads to the evolution
\begin{align}
\dot{V}_{1,n}&=-\bm{z}_0^\top Q_n \bm{z}_0-\z_1^\top G_{z_1} \z_1+\z_0^\top P_n B \bm{\rho}_n(\x)\notag\\
&+\z_1^\top\Big(\Big[ \frac{\partial \hat{\f}_n}{\partial \x}+G_n\Big]B\bm{\rho}_n(\x)+\dot{\g}-\g_{\dot{d}}\Big).
\end{align}
Next, we define the error $\z_2\in\R^3$ with
\begin{align}
\bm{z}_2(t,n,\x,R,u)=\dot{\g}(R,u)-\g_{\dot{d}}(t,n,\x,R,u),
\end{align}
and an extended Lyapunov function
\begin{align}
V_{n}(\z_0,\z_1,\z_2)=V_{1,n}+\frac{1}{2}\z_2^\top \z_2\geq 0.\label{for:V2}
\end{align}
The derivative of $V_n$ leads to
\begin{align}
\dot{V}_{n}\!=\!\dot{V}_{1,n}\!+\!\bm{z}_2^\top\Big(&\ddot{\g}\!-\!m\bm{p}_d^{(4)}\!+\!G_n\ddot{\z}_0\!+\!BP_n\dot{\z}_0\!+\!\frac{d}{dt}\Big[\frac{\partial {\hat{\f}}_n}{\partial \x}\dot{\hat{\x}}\Big]\Big)\notag
\end{align}
and we construct a desired value of $\ddot{\g}$ with $\g_{\ddot{d}}$ given by~\cref{for:gddd}. Then, it is substituted into $\dot{V}_{n}$ to obtain
\begin{align}
\dot{V}_{n}&= -\bm{z}_0^\top Q_n \bm{z}_0-\z_1^\top G_{z_1}\z_1-\z_2^\top G_{z_2}\z_2\label{for:UPV2}\\
&+\!(\z_0^\top P_n\!+\!\z_1^\top D(\x) \!+\!\z_2^\top E(\x))B\bm{\rho}_n(\x)\!+\!\z_2^\top(\ddot{\g}\!-\!\g_{\ddot{d}})\notag\\
D(\x)&\coloneqq \frac{\partial \hat{\f}_n}{\partial\x}+G_n\\
E(\x)&\coloneqq BP+G_{z_1}D+G_n\frac{\partial \dot{\hat{\x}}}{\partial \x}+\frac{\partial}{\partial \x}\Big(\frac{\partial{\hat{\f}}_n}{\partial \x}\dot{\hat{\x}}\Big).
\end{align}
To eliminate the last summand in~\cref{for:UPV2} except for the estimation error $\bm{\rho}_\omega$, we note that
\begin{align}
\ddot{\g}(R,u)=R(\check{\bm\omega}^2\bm{e}u+2\check{\bm\omega}\dot{u}+\dot{\bm\omega}\times \bm{e}u+\bm{e}\ddot{u}),
\end{align}
such that $\ddot{\g}-\g_{\ddot{d}}=-RJ^{-1}\bm{\rho}_\omega(\bm{s})\times\bm{g}$ for
\begin{align}
\dot{\bm\omega}\times \bm{e}u+\bm{e}\ddot{u}=R^\top\g_{\ddot{d}}-\check{\bm\omega}^2\bm{e}u-2\check{\bm\omega}\dot{u}\label{for:inpu}.
\end{align}
Using~\cref{for:inpu} and~\cref{ass:1}, the evolution of the Lyapunov function $V$ can be upper bounded by
\begin{align}
\Prob\{\dot{V}_{n}&\leq -\bm{z}_0^\top Q_n \bm{z}_0-\z_1^\top G_{z_1}\z_1-\z_2^\top G_{z_2}\z_2\label{for:UPV3}\\
&+\Verts{(\z_0^\top P_n+\z_1^\top \bar{D} +\z_2^\top \bar{E})B+\bar{c}\z_2}\bar{\rho}_n(\bm{s})\}\geq\delta,\notag
\end{align}
with the upper bounds $\bar{D},\bar{E}\in\R^{3\times 3}$ and $\bar{c}\in\R$, which exist due to~\cref{ass:1}. Thus, the evolution is negative with probability $\delta$ for all $\z=[\z_1^\top,\z_2^\top,\z_3^\top]^\top$ outside a ball
\begin{align}
\Verts{\z}>\max_{\bm{s}\in S_\X}\bar{\rho}_n(\bm{s})\underbrace{\frac{\Verts{P_n B}+\Verts{\bar{D} B}+\Verts{\bar{E} B}+\bar{c}}{\min\{\eig(Q_n),\eig(G_{z_1}),\eig(G_{z_2})\}}}_{\eqqcolon\lambda_n},\notag
\end{align}
where a maximum of $\bar{\rho}_n$ exists regarding to~\cref{ass:1}. Finally, the Lyapunov function~\cref{for:V2} is lower and upper bounded by $\alpha_1(\Verts{\z})\leq V_n(\z)\leq\alpha_2(\Verts{\z})$, where $\alpha_1(r)=0.5\min\{\eig(P_n),1\}r^2$ and $\alpha_2(r)=0.5\max\{\eig(P_n),1\}r^2$. Thus, we can compute the radius $b_n\in\R_{\geq_0}$ of the bound by
\begin{align}
b_n=\sqrt{\frac{\max\{\eig(P_n),1\}}{\min\{\eig(P_n),1\}}}.\label{for:bound}
\end{align}
Since~\cref{ass:2} only allows a finite number of switches, there exists a time $T\in\R_{\geq 0}$ such that $n(t)=n_{\text{end}}\in\N$ for all $t\geq T$. Thus, $\Prob\{\Verts{\z_0(t)}\leq \max_{\bm{s}\in S_\X}\bar{\rho}_{n_\text{end}}(\bm{s})b_{n_\text{end}},\forall t\geq T\}\geq\delta$.
\end{proof}
\begin{rem}
Extension to the rotation are analogously to perform with additional terms in the Lyapunov function as given in~\cite{mahony2004robust,frazzoli2000trajectory}.
\end{rem}
\begin{rem}
The torque control law of~\cref{for:ctrllaw} has a singularity at $u=0$ as without control force $u$ no tracking control is possible in general. To overcome the singularity, a reasonable trajectory planning can be performed, see~\cite{mahony2004robust} or the control torques are set to zero at this point. In practice, this leads
to chattering that can be alleviated by a slight modification of the control law to remove the singularity, see~\cite{khalil1996noninear}.
\end{rem}
\begin{rem}
The proposed approach allows multiple ways of data collection and adaptation of the feedback matrix $G_n$. A possible strategy can be time-triggered where new data points are recurrently attached to the data set $\D_n$ to improve the prediction accuracy of the oracle and the magnitude of $G_n$ is decreased over time. More advanced strategies can be based on the model uncertainty or the tracking error as used in, e.g.,~\cite{umlauft:TAC2020}.
\end{rem}
The proof shows that the bound of the tracking error~\cref{for:bound} depends on the prediction error $\bar{\rho}_n$ of the oracle. 
\section{Numerical example}\label{sec:num}
In this section, we present a numerical example of a quadrocopter within an a-priori unknown wind field. The dynamics of the quadrocopter are described by~\cref{for:system} with mass $m=\SI{1}{\kilogram}$, inertia $J=\diag(2,2,1)\si{\kilogram\metre^2}$ and the direction $\bm{e}=[0,0,1]^\top$ of the force input $u$. As unknown dynamics~$\f$ and~$\f_\omega$, we consider an arbitrarily chosen wind field and the gravity force given by
\begin{align}
\f(\x)&=[0,0,2\sin(x_1)+\exp(-5x_2^2)-9.81]^\top\\
\f_\omega(\bm{s})&=[2\exp(-x_1^2-x_2^2)+\omega_1\cos(x_2)^2,0,0]^\top.
\end{align}
A GP model is used as oracle to predict the z-component of $\f(\x)$ and the x-component of $\f_\omega(\bm{s})$ with the squared exponential kernel, see~\cite{rasmussen2006gaussian}. The prior knowledge about the existing gravity in $\f(\x)$ is considered as estimate in the mean function of the GP with $m_3(\bm{s})=-10$. At starting time $t=0$, the data set $\D_n$ is empty such that the prediction is solely based on the mean function. The initial position of the quadrocopter is $\bm{p}(0)=[0.1,-0.1,0]^\top$ whereas the desired trajectory starts at $\bm{p}_d(0)=[0,0,0]^\top$ due to an assumed position measurement error. In this example, we employ an online learning approach which collects a new training point every $\SI{0.1}{\second}$ such that the total number of training points is $N=5n$. In \cref{fig:tra}, the first $\SI{3}{\second}$ of the desired (dashed) and the actual trajectory (solid) is shown. The crosses denote the collected training data. Each training point consists of the actual state $\bm{s}$ and $\bm{y}$ as given by~\cref{for:dataset}. Since the training point depends on the typically noisy measurement of the accelerations $\ddot{\bm{p}}$ and $\dot{\bm{\omega}}$, Gaussian distributed noise~$\mathcal{N}(0,0.08^2I_3)$ is added to the measurement. The GP model is updated every second until $t=\SI{12}{\second}$, where the last 10 collected training points are appended to the set $\D_n$ and the hyperparameters are optimized by means of the likelihood function, see~\cite{steinwart2008support}. Thus, the function $n$ is the integer part of $2t$ up to $t=\SI{12}{\second}$ given by $n(t)=\min(\SI{12}{\second},\left\lfloor 2t\right\rfloor)$. The initial feedback gain matrix is set to $G_{n=0}=[\diag(10,10,40),\diag(10,10,10)]$ and $G_{z_1}=G_{z_2}=2I_3$. 
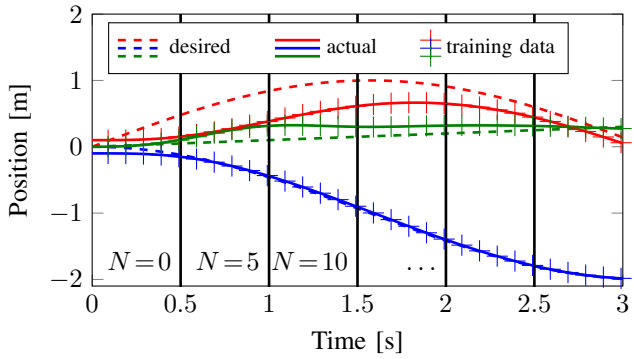
\begin{figure}[thb]
\begin{center}
\tikzsetnextfilename{tra}
	\begin{tikzpicture}
\begin{axis}[
  xlabel={Time [s]},
  ylabel={Position [m]},
  legend pos=north west,
  width=\columnwidth,
  height=5.2cm,
  ymin=-2.1,
  ymax=2,
  xmin=0,
  xmax=3,
  legend style={font=\footnotesize},
  legend cell align={left},
  legend style={row sep=-7pt},
  legend style={/tikz/every even column/.append style={column sep=0.5cm}},
  legend columns=3,transpose legend]
\addplot[color=red,dashed,line width=1pt,no marks] table [x index=0,y index=1]{data/trajectory.dat};
\addplot[color=blue,dashed,line width=1pt,no marks] table [x index=0,y index=2]{data/trajectory.dat};
\addplot[color=green!50!black,dashed,line width=1pt,no marks] table [x index=0,y index=3]{data/trajectory.dat};
\addplot[color=red,line width=1pt,no marks] table [x index=0,y index=4]{data/trajectory.dat};
\addplot[color=blue,line width=1pt,no marks] table [x index=0,y index=5]{data/trajectory.dat};
\addplot[color=green!50!black,line width=1pt,no marks] table [x index=0,y index=6]{data/trajectory.dat};
\addplot[color=red,only marks,mark=+,mark size=4pt] table [x index=0,y index=1]{data/data.dat};
\addplot[color=blue,only marks,mark=+,mark size=4pt] table [x index=0,y index=2]{data/data.dat};
\addplot[color=green!50!black,only marks,mark=+,mark size=4pt] table [x index=0,y index=3]{data/data.dat};
\addplot[color=black,line width=1pt,no marks] coordinates {(0.5,-5) (0.5,5)};
\addplot[color=black,line width=1pt,no marks] coordinates {(1,-5) (1,5)};
\addplot[color=black,line width=1pt,no marks] coordinates {(1.5,-5) (1.5,5)};
\addplot[color=black,line width=1pt,no marks] coordinates {(2,-5) (2,5)};
\addplot[color=black,line width=1pt,no marks] coordinates {(2.5,-5) (2.5,5)};
\node[anchor=south east] at (axis cs:0.5,-2) {$N\!=\!0$};
\node[anchor=south east] at (axis cs:1,-2) {$N\!=\!5$};
\node[anchor=south east] at (axis cs:1.5,-2) {$N\!=\!10$};
\node[anchor=south east] at (axis cs:2,-2) {$\ldots$};
\legend{\phantom{a},desired,\phantom{a},\phantom{a},actual,\phantom{a},\phantom{a},training data,\phantom{a}};
\end{axis}
\end{tikzpicture} 
	\vspace{-0.5cm}
	\caption{A segment of the desired (dashed) and actual trajectory (solid). Every $\SI{0.1}{\second}$ a training point (cross) is recorded. Every $\SI{0.5}{\second}$ the oracle is updated based on all collected training points $N$ up to this point. The additional training data allows to refine the model such that the tracking error is decreasing.}
	\vspace{-0.2cm}
	\label{fig:tra}
\end{center}
\end{figure}
In this example, we adapt the feedback gain matrix based on the number of training points. When the GP model is updated with new training data, the feedback gains are decreased by $G_{n}=0.9^n G_{n=0}$. Thus, after the first update, the feedback gains are $90\%$ of the initial gains, see~\cref{fig:VK}. The simulation time is $\SI{14}{\second}$.~\Cref{fig:tra3d} visualizes that the actual position (solid) of the quadrocopter converges to a very tight set around the desired position (dashed). The effects of the switching to the updated GP model are more noticeable in the evolution of the Lyapunov function in~\cref{fig:VK}. The function might increase after an update of the GP model due to the change of $G_n$ and the new prediction accuracy of the GP model. However, the function converges to a bounded set as proposed in~\cref{thm:1} after the finite number of switching events.
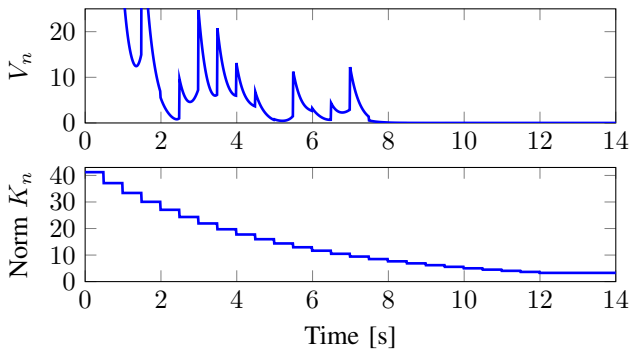
\begin{figure}[thb]
\begin{center}
	\tikzsetnextfilename{VK1}
\begin{tikzpicture}
\begin{axis}[
name=plot1,
  ylabel={$V_n$},
  legend pos=north west,
  width=\columnwidth,
  height=3.1cm,
  restrict y to domain=-1:100,
  ymin=0,
  ymax=25,
  xmin=0,
  xmax=14,
  legend style={font=\footnotesize},
  legend cell align={left}]
\addplot[color=blue,line width=1pt,no marks] table [x index=0,y index=1]{data/VK.dat};
\end{axis}
\begin{axis}[
name=plot2,
   at=(plot1.below south east), anchor=above north east,
  xlabel={Time [s]},
  ylabel={Norm $K_n$},
  legend pos=north west,
  width=\columnwidth,
  height=3.1cm,
  ymin=0,
  ymax=43,
  xmin=0,
  xmax=14,
  legend style={font=\footnotesize},
  legend cell align={left}]
\addplot[color=blue,line width=1pt,no marks] table [x index=0,y index=2]{data/VK.dat};
\end{axis}
\end{tikzpicture} 
	\vspace{-0.7cm}
	\caption{Top: Lyapunov function converges to a tight set around zero. The jumps occur when the oracle is updated. Bottom: Norm of the feedback gain matrix is decreasing due to improved accuracy of the oracle.}
	\vspace{-0.5cm}
	\label{fig:VK}
\end{center}
\end{figure}
\begin{figure}[thb]
\begin{center}
\tikzsetnextfilename{tra3d}
	\begin{tikzpicture}
\begin{axis}[
  xlabel={x-position [m]},
  ylabel={y-position [m]},
  zlabel={z-position [m]},
  xlabel shift = -10 pt,
  ylabel shift = -10 pt,
  legend pos=north west,
  width=\columnwidth,
  height=6.5cm,
  legend style={font=\footnotesize},
  legend cell align={left},
  view={-30}{20},
  legend style={at={(0,1.01)},anchor=north west,/tikz/every even column/.append style={column sep=0.2cm}},
  legend columns=3]
\addplot3[color=red,dashed,line width=1pt,no marks] table [x index=1,y index=2,z index=3]{data/trajectory.dat};
\addplot3[color=blue,line width=1pt,no marks] table [x index=4,y index=5,z index=6]{data/trajectory.dat};
\addplot3[only marks,color=green!50!black,line width=1pt,mark=o] coordinates {(0.1,-0.1,0)};
\addplot3[only marks,color=green!50!black,line width=1pt,mark=o] coordinates {(0,0,0)};
\legend{desired,actual,initial state};
\end{axis}
\end{tikzpicture} 
	\vspace{-0.4cm}
	\caption{Actual trajectory converges to desired trajectory.}
	\vspace{-0.5cm}
	\label{fig:tra3d}
\end{center}
\end{figure}
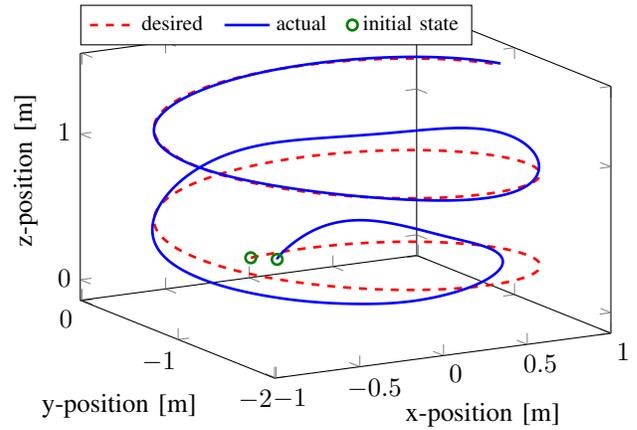
\section*{Conclusion}
We present an online learning-based tracking control law for a class of underactuated systems with unknown dynamics typical for aerial and underwater vehicles. Using a various type of oracles, the tracking error is proven to be bounded in probability and the size of the bound is explicitly given. The online fashion of the proposed approach allows to improve the quality of the oracle over time and, thus, to improve the tracking performance. Even though no particular oracle is assumed, we show that Gaussian process models fulfill all requirements to be used as oracle in the proposed control scheme. Finally, a numerical example visualizes the effectiveness of the control law.

\bibliography{mybib}

\begin{thebibliography}{10}

\bibitem{reyhanoglu1999dynamics}
M.~Reyhanoglu, A.~van~der Schaft, N.~H. McClamroch, and I.~Kolmanovsky,
  ``Dynamics and control of a class of underactuated mechanical systems,'' {\em
  IEEE Transactions on Automatic Control}, vol.~44, no.~9, pp.~1663--1671,
  1999.

\bibitem{lee2009feedback}
D.~Lee, H.~J. Kim, and S.~Sastry, ``Feedback linearization vs. adaptive sliding
  mode control for a quadrotor helicopter,'' {\em International Journal of
  control, Automation and systems}, vol.~7, no.~3, pp.~419--428, 2009.

\bibitem{raffo2008backstepping}
G.~V. Raffo, M.~G. Ortega, and F.~R. Rubio, ``Backstepping/nonlinear
  ${H}_\infty$ control for path tracking of a quadrotor unmanned aerial
  vehicle,'' in {\em Proc. of the American Control Conference}, pp.~3356--3361,
  2008.

\bibitem{frazzoli2000trajectory}
E.~Frazzoli, M.~A. Dahleh, and E.~Feron, ``Trajectory tracking control design
  for autonomous helicopters using a backstepping algorithm,'' in {\em Proc. of
  the American Control Conference}, pp.~4102--4107, IEEE, 2000.

\bibitem{hou2013model}
Z.-S. Hou and Z.~Wang, ``From model-based control to data-driven control:
  Survey, classification and perspective,'' {\em Information Sciences},
  vol.~235, pp.~3--35, 2013.

\bibitem{mahony2004robust}
R.~Mahony and T.~Hamel, ``Robust trajectory tracking for a scale model
  autonomous helicopter,'' {\em International Journal of Robust and Nonlinear
  Control: IFAC-Affiliated Journal}, vol.~14, no.~12, pp.~1035--1059, 2004.

\bibitem{kobilarov2013trajectory}
M.~Kobilarov, ``Trajectory tracking of a class of underactuated systems with
  external disturbances,'' in {\em 2013 American Control Conference},
  pp.~1044--1049, IEEE, 2013.

\bibitem{umlauft:TAC2020}
J.~Umlauft and S.~Hirche, ``Feedback linearization based on {Gaussian}
  processes with event-triggered online learning,'' {\em IEEE Transactions on
  Automatic Control}, 2020.

\bibitem{greef2021letters}
M.~{Greeff} and A.~P. {Schoellig}, ``Exploiting differential flatness for
  robust learning-based tracking control using {Gaussian} processes,'' {\em
  IEEE Control Systems Letters}, vol.~5, no.~4, pp.~1121--1126, 2021.

\bibitem{beckers2019automatica}
T.~Beckers, D.~Kulić, and S.~Hirche, ``Stable {Gaussian} process based
  tracking control of {Euler}-{Lagrange} systems,'' {\em Automatica}, no.~103,
  pp.~390--397, 2019.

\bibitem{helwa-ral19}
M.~K. Helwa, A.~Heins, and A.~P. Schoellig, ``Provably robust learning-based
  approach for high-accuracy tracking control of {L}agrangian systems,'' {\em
  {IEEE Robotics and Automation Letters}}, vol.~4, no.~2, pp.~1587--1594, 2019.

\bibitem{han2021stable}
F.~Han and J.~Yi, ``Stable learning-based tracking control of underactuated
  balance robots,'' {\em IEEE Robotics and Automation Letters}, vol.~6, no.~2,
  pp.~1543--1550, 2021.

\bibitem{berkenkamp-icra16}
F.~Berkenkamp, A.~P. Schoellig, and A.~Krause, ``Safe controller optimization
  for quadrotors with {G}aussian processes,'' in {\em {Proc. of the IEEE
  International Conference on Robotics and Automation (ICRA)}}, pp.~491--496,
  May 2016.

\bibitem{scarselli1998universal}
F.~Scarselli and A.~C. Tsoi, ``Universal approximation using feedforward neural
  networks: A survey of some existing methods, and some new results,'' {\em
  Neural networks}, vol.~11, no.~1, pp.~15--37, 1998.

\bibitem{rasmussen2006gaussian}
C.~E. Rasmussen and C.~K. Williams, {\em {Gaussian} processes for machine
  learning}, vol.~1.
\newblock MIT press Cambridge, 2006.

\bibitem{steinwart2008support}
I.~Steinwart and A.~Christmann, {\em Support vector machines}.
\newblock Springer Science \& Business Media, 2008.

\bibitem{aastrom1971system}
K.~J. {\AA}str{\"o}m and P.~Eykhoff, ``System identification—a survey,'' {\em
  Automatica}, vol.~7, no.~2, pp.~123--162, 1971.

\bibitem{beckers:ecc2016}
T.~Beckers and S.~Hirche, ``Stability of {Gaussian} process state space
  models,'' in {\em Proc. of the European Control Conference}, 2016.

\bibitem{wolpert1996lack}
D.~H. Wolpert, ``The lack of a priori distinctions between learning
  algorithms,'' {\em Neural computation}, vol.~8, no.~7, pp.~1341--1390, 1996.

\bibitem{wahba1990spline}
G.~Wahba, {\em Spline models for observational data}.
\newblock SIAM, 1990.

\bibitem{srinivas2012information}
N.~Srinivas, A.~Krause, S.~M. Kakade, and M.~W. Seeger, ``Information-theoretic
  regret bounds for {Gaussian} process optimization in the bandit setting,''
  {\em IEEE Transactions on Information Theory}, vol.~58, no.~5,
  pp.~3250--3265, 2012.

\bibitem{liberzon1999basic}
D.~Liberzon and A.~S. Morse, ``Basic problems in stability and design of
  switched systems,'' {\em IEEE control systems magazine}, vol.~19, no.~5,
  pp.~59--70, 1999.

\bibitem{khalil1996noninear}
H.~K. Khalil, {\em Noninear Systems}.
\newblock Prentice-Hall, New Jersey, 1996.

\end{thebibliography}
\bibliographystyle{ieeetr}

\end{document}